\documentclass[aps, pra, 10pt, amsmath, amssymb, notitlepage, twocolumn]{revtex4-1}
\usepackage{mathtools}
\usepackage{algcompatible}

\usepackage[vcentermath]{youngtab}
\usepackage{qcircuit}

\usepackage{etoolbox}
\usepackage{appendix}
\usepackage{lineno}

\usepackage{graphicx}   
\usepackage[utf8]{inputenc}
\usepackage[main=english, czech]{babel}
\usepackage[normalem]{ulem}
\usepackage{bm}
\usepackage{natbib}
\usepackage{amsthm}
\usepackage{braket}
\usepackage{proof}
\newtheorem{definition}{Definition}

\newtheorem{lemma}{Lemma}

\usepackage{url}
\usepackage{calc}
\usepackage{tikz-qtree}
\usetikzlibrary{trees, patterns}
\usetikzlibrary{positioning}
\usepackage{tikz,ifthen,calc}
\newtheorem{theorem}{Theorem}
\usepackage{complexity}
\usepackage{fancyref}
\usetikzlibrary{shapes}

\definecolor{ltgray}{rgb}{0.95,0.95,0.95}
\pgfdeclarepatternformonly{soft horizontal lines}{\pgfpointorigin}{\pgfqpoint{100pt}{1pt}}{\pgfqpoint{100pt}{3pt}}%
{
  \pgfsetstrokeopacity{0.3}
  \pgfsetlinewidth{0.1pt}
  \pgfpathmoveto{\pgfqpoint{0pt}{0.5pt}}
  \pgfpathlineto{\pgfqpoint{100pt}{0.5pt}}
  \pgfusepath{stroke}
}

\usetikzlibrary{arrows}

\tikzset{
  treenode/.style = {align=center, inner sep=0pt, circle, text centered, minimum size=0.4cm, font=\fontsize{8}{22.4}\sffamily },
  youngnode/.style = {align=center, inner sep=0pt, text centered, minimum size=0.4cm, font=\fontsize{2}{10.4}\sffamily },
  arn_n/.style = {treenode, circle, white, draw=black,
    fill=black, text width=1em},
  arn_r/.style = {treenode, circle, white, draw=red, fill=red, 
    text width=1em, thick},
  arn_x/.style = {treenode, rectangle, draw=black,
    minimum width=0.5em, minimum height=0.5em}

}

\usepackage{hyperref}
\hypersetup{
    colorlinks=true,
    linkcolor=blue,  
    citecolor=blue,
    urlcolor=blue,
}
 
\usepackage{comment}
\usepackage{algorithm2e}

\begin{document}
\title{A Classical Algorithm for Quantum $\textsf{SU}(2)$ Schur Sampling}
%---------------------------------------------------

\author{\foreignlanguage{czech}{Vojtěch Havlíček}}
\email{vojtech.havlicek@keble.ox.ac.uk}
\affiliation{Department of Computer Science, University of Oxford, Wolfson Building, Parks Road, Oxford
OX1 3QD, UK}  

\author{Sergii Strelchuk}
\affiliation{Department of Applied Mathematics and Theoretical Physics, University of Cambridge, Wilberforce Road,  Cambridge, CB2 3HU, UK} 

\author{Kristan Temme}
\affiliation{IBM T.J. Watson Research Center, Yorktown Heights, NY 10598, USA}

\begin{abstract} 
Many quantum algorithms can be represented in a form of a classical circuit positioned between quantum Fourier transformations. Motivated by the search for new quantum algorithms, we turn to circuits where the latter transformation is replaced by the $\textsf{SU}(2)$ quantum Schur Transform -- a global transformation which maps the computational basis to a basis defined by angular momenta. We show that the output distributions of these circuits can be approximately classically sampled in polynomial time if they are sufficiently close to being sparse, thus isolating a regime in which these Quantum $\textsf{SU}(2)$ Schur Circuits could lead to algorithms with exponential computational advantage. %by sampling.
Our work is primarily motivated by a conjecture that underpinned the hardness of Permutational Quantum Computing, a restricted quantum computational model that has the above circuit structure in one of its computationally interesting regimes. The conjecture stated that approximating transition amplitudes of Permutational Quantum Computing model to inverse polynomial precision on a classical computer is computationally hard. We disprove the extended version of this conjecture -- even in the case when the hardness of approximation originated from a difficulty of finding the large elements in the output probability distributions. 
Finally, we present some evidence that output of the above Permutational Quantum Computing circuits could be efficiently approximately sampled from on a classical computer. 
\end{abstract}
\maketitle

\section{Introduction}

%%%% PAR 1: BROADER CONTEXT
Charaterizing the power of quantum computers is one of the two major challenges in quantum computation, with the other being their scalable implementation. A seminal approach to the former problem is the study of conditions which make quantum algorithms amenable to methods of efficient classical simulation. 
A number of important quantum algorithms can be cast in a form of classical circuit positioned between a pair of circuits which implement quantum Fourier transformation. These are, for example, algorithms for the Hidden Subgroup Problem which in particular include the Shor's factoring algorithm \cite{Shor94,Kitaev95}. While the latter provides strong evidence that quantum computers outperform the classical ones, Schwarz and van den Nest \cite{Schwarz13} showed that the respective quantum circuit could be efficiently classically simulated if its output distribution was sufficiently close to being sparse. 
\begin{figure}[t]
\includegraphics[scale=.85]{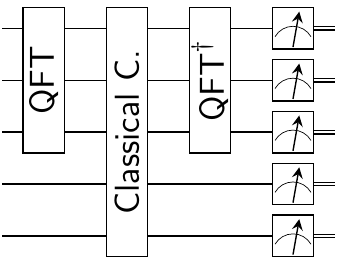}~\hspace{.5cm}
\includegraphics[scale=.85]{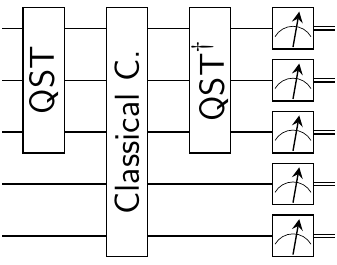}
\caption{Schematic diagrams of the quantum circuit used in Shor's factoring algorithm (Left) and the circuits we consider here (Right). QST denotes the $\textsf{SU}(2)$ Quantum Schur Transformation. The classical circuits between the transforms can represent, for example a polynomially-long sequence of Toffoli gates.}
\label{Fig:Circuits}
\end{figure}

In our current work, we aim to characterize a { different} class of circuits that instead of the quantum Fourier transform contain the quantum Schur transform (QST) as depicted on Fig.~\ref{Fig:Circuits}. QST is a map from the computational basis to a basis defined by angular momentum \cite{Harrow05, Bacon06, Bacon06b} and it underpins a variety of quantum information processing tasks, including spectrum estimation \cite{Gill00, Keyl01}, hypothesis testing \cite{Hayashi01,Hayashi02,Hayashi02b,Hayashi02c}, quantum computing using decoherence-free subspaces \cite{Kempe01}, communication without a shared reference frame \cite{BRS03, BRS07}, and quantum color coding \cite{HHH05}. A quantum circuit that efficiently implements this transform was first described in \cite{Harrow05,Bacon06,Bacon06b} and recently improved by Kirby and Strauch \cite{Kirby17, Kirby17b}. 
The extent to which circuits using QST could be used to devise new quantum algorithms is, to our knowledge, largely unexplored - possibly with the exception of \cite{Childs06} and \cite{Jordan09}. 

\begin{figure*}[t]
\includegraphics[scale=0.5]{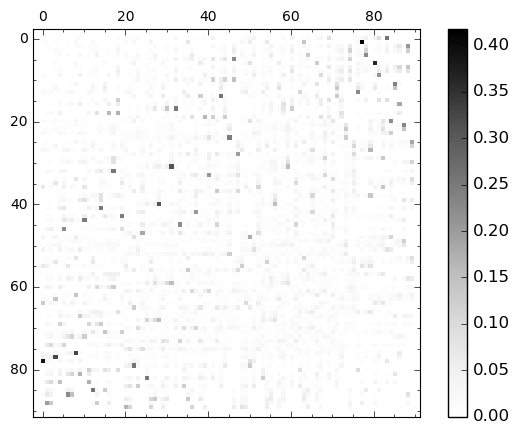}~\hspace{2cm}
\includegraphics[scale=0.5]{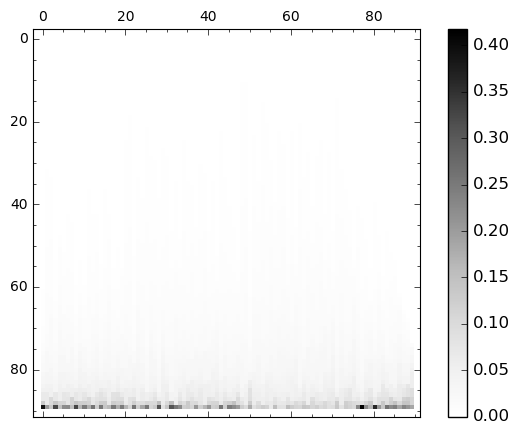}
\caption{(Left) A part of the $|\braket{x|C|y}|^2$ matrix for a typical PQC instance. After normalization, most matrix elements are indistinguishable from zeros within the polynomially small approximation window. We show how to classically find the large elements. (Right) The output matrix with sorted output, demonstrating that an overwhelming fraction of the probabilities are usually small.}
\label{Fig:PQC-SEQ}
\end{figure*}

QST is a centerpiece in the analysis of Permutational Quantum Computing (PQC) \cite{Havlicek18} -- a restricted quantum computational model based on recoupling of angular momenta~\cite{Jordan09,Marzuoli05}. It has been conjectured that PQC has supra-classical computational power. One of the conjectures supporting this belief stated that an approximation of its transition amplitudes in the regime where they encode matrix elements of the symmetric group irreps in the Young's orthogonal form~\cite{Jordan08, Jordan09} is hard to compute classically if we require inverse polynomial precision (in the number of input qubits). While in our previous work we presented an efficient classical algorithm for approximating such transition amplitudes~\cite{Havlicek18}, an intriguing question remained: Is it also possible to identify all PQC transition amplitudes that can be approximated using classical methods with the inverse polynomial precision? Since the expected output probability of an $n$-qubit quantum circuit $C$ with an input state $\ket{y}$ is given by:
\begin{align*}
\mathbb{E}_x\left(|\braket{x|C|y}|^2\right) &=  \frac{1}{2^n} \sum_x |\braket{x|C|y}|^2 = \frac{1}{2^n},
\end{align*}
approximating these values with an inverse polynomial precision cannot distinguish the majority of $\braket{x|C|y}$ amplitudes from zeroes (see Fig.~\ref{Fig:PQC-SEQ}). 
Could we exploit the difficulty that arises from finding large matrix elements encoded in the output of the algorithm and thus demonstrate the (exponential) quantum computational advantage? 

We show that this is not the case by describing a classical method that finds all large output probabilities in polynomial time. 
Our proof technique uses the simulation technique of Schwarz and van den Nest~\cite{Schwarz13} where the authors studied analogous problem in the context of the quantum Fourier transform. This approach uses a variant of the Kushilevitz-Mansour algorithm used in classical learning theory \cite{kushilevitz1993learning,Goldreich89}. We adapt it for distributions arising in the class of circuits using the QST, which include the relevant regime of Permutational Quantum Computing.  We then show how to classically approximately sample their output distributions. The sampling algorithm becomes efficient for output distributions that are sufficiently close to sparse.

Our results additionally imply that sampling from the quantum Schur circuits can only lead to exponential computational advantage if the individual elements of the output distribution \textit{cannot} be resolved by polynomial approximation with the quantum device by taking polynomially many samples. A way to circumvent this restriction, similarly to the case of circuits that use the quantum Fourier transform, could be to use a technique utilized in the Shor's algorithm that reconstructs group generators by sampling $\log |G|$ group elements for a super-polynomially large $|G|$. There is no meaningful counterpart to this approach for the QST as of now.

\section{Quantum $\textsf{SU}(2)$ Schur Sampling}

\begin{figure}[t]
\begin{tikzpicture}[scale = 1.3]
    \node[circle, inner sep=2pt] at(-.65, 0.75) {$1$};
    \node[circle, inner sep=2pt] at(.65, 0.75) {$3$};
    \node[circle, inner sep=2pt] at(1.3, 0.75) {$n-1$};
    \node[circle, inner sep=2pt] at(1.95, 0.75) {$n$};
    \node[circle, inner sep=2pt] at(-0.0, 0.75) {$2$};
    \node at(-.55, -0.15) {$S^2_{[2]}$};
    \node at(-.15, -0.55) {$S^2_{[3]}$};
    \node at(.15, -1.1) {$S^2_{[n-1]}$};
    \node at(.6, -1.8) {$S^2_{[n]}, Z$};
    
    \draw[color=gray,thick] (-0.65, 0.6) -- (-0.3, 0.1); 
    \draw[color=gray,thick] (-0, 0.6) -- (-0.3, 0.1); 
    \draw[color=gray,thick] (0.65, 0.6) -- (0.3, 0.1); 
    \draw[color=gray,thick] (1.3, 0.6) -- (0.3, -.7); 
     \draw[color=gray,thick] (1.95, 0.6) -- (0.6, -1.1); 
    \draw[color=gray,thick] (-0.3,0.1)  -- (0, -.3);
    \draw[color=gray,thick, dashed] (0,-0.3)  -- (0.3, -.7); %
    \draw[color=gray,thick] (0.3,-.7)  -- (0.6, -1.1);%
    \draw[color=gray,thick] (0.3,0.1) -- (0, -0.3);
    \draw[color=gray,thick] (0.6, -1.1) -- (0.6, -1.5);
    
     \node[circle, fill = black,inner sep=1pt] at (-0.3, 0.1){};
    \node[circle, fill = black,inner sep=1pt] at (0, -0.3){};
    \node[circle, fill = black,inner sep=1pt] at (0.3, -0.7){};
    \node[circle, fill = black,inner sep=1pt] at (0.6, -1.1){};
\end{tikzpicture}
\caption{Sequentially coupled basis on $n$ qubits. The numbers at the leaf nodes label qubits. Every vertex $\bullet$ carries a total spin operator $S^2_A$, that forces qubits in set $A$ to one of its eigenstates.
Similar diagrams can be used to label basis states and are shown in Appendix~\ref{Appendix:StateDiagrams} or \cite{Jordan09}.}
\label{Fig:SequentiallyCoupled}
\end{figure}
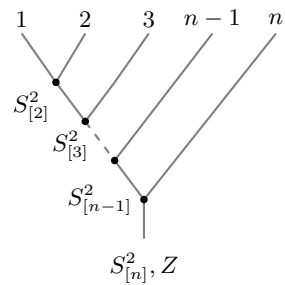

\begin{figure*}[t]
\begin{tikzpicture}[scale=.8]
\draw [->](.75,0)--(.75,5) node[left]{$J$};
\draw [->](.75,0)--(5.5,0) node[right]{$n$};
\foreach \x/\xtext in {0.5/1, 1/2, 1.5/3, 2/4, 2.5/5, 3/6}
{\draw (1.5*\x cm,1pt ) -- (1.5*\x cm,-1pt ) node[anchor=north] at(1.5*\x cm, -0.3) {$\xtext$};}

\node[treenode] at(0.75,0.75) (1) {$\frac{1}{2}$};
\node[treenode] at(1.5,1.5) (2) {$1$};
\node[treenode] at(2.25,2.25) (3) {$\frac{3}{2}$};
\node[treenode] at(3,3) (4) {$2$};
\node[treenode] at(3.75,3.75) (5) {$\frac{5}{2}$};
\node[treenode] at(4.5,4.5) (6) {$3$};

\node[treenode] at(1.5,0) (7) {$0$};
\node[treenode] at(2.25,0.75) (8) {$\frac{1}{2}$};
\node[treenode] at(3,1.5) (9) {$1$};
\node[treenode] at(3.75,2.25) (10) {$\frac{3}{2}$};
\node[treenode] at(4.5,3) (11) {$2$};

\node[treenode] at(3,0) (12) {$0$};
\node[treenode] at(3.75,.75) (13) {$\frac{1}{2}$};
\node[treenode] at(4.5,1.5) (14) {$1$};

\node[treenode] at(4.5,0) (15) {$0$};

\node[] at (0,0) (0) {};

\draw[->] (1) edge (2) (2) edge (3) (3) edge (4) (4) edge (5)(5) edge (6);
\draw[->,thick] (7) edge (8) (8) edge (9); 
\draw[->] (9) edge (10) (10) edge (11);
\draw[->] (12) edge (13) (13) edge (14);
\draw[->,thick] (1) edge (7);
\draw[->] (2) edge (8)(8) edge (12); 
\draw[->] (3) edge (9)(9) edge (13) (13) edge (15);
\draw[->] (4) edge (10)(10) edge (14);
\draw[->] (5) edge (11);
\draw[->] (6);
\end{tikzpicture}
\hspace{3cm}
\begin{tikzpicture}[scale=.8]
\draw [->](.75,0)--(.75,5) node[left]{$J$};
\draw [->](.75,0)--(5.5,0) node[right]{$n$};
\foreach \x/\xtext in {0.5/1, 1/2, 1.5/3, 2/4, 2.5/5, 3/6}
{\draw (1.5*\x cm,1pt ) -- (1.5*\x cm,-1pt ) node[anchor=north] at(1.5*\x cm, -0.3) {$\xtext$};}

\node[youngnode] at(0.75,0.75) (1) {$\yng(1)$};
\node[youngnode] at(1.5,1.5) (2) {$\yng(2)$};
\node[youngnode] at(2.25,2.25) (3) {$\yng(3)$};
\node[youngnode] at(3,3) (4) {$\yng(4)$};
\node[youngnode] at(3.75,3.75) (5) {$\yng(5)$};
\node[youngnode] at(4.5,4.5) (6) {$\yng(6)$};

\node[youngnode] at(1.5,0) (7) {$\yng(1,1)$};
\node[youngnode] at(2.25,0.75) (8) {$\yng(2,1)$};
\node[youngnode] at(3,1.5) (9) {$\yng(3,1)$};
\node[youngnode] at(3.75,2.25) (10) {$\yng(4,1)$};
\node[youngnode] at(4.5,3) (11) {$\yng(5,1)$};

\node[youngnode] at(3,0) (12) {$\yng(2,2)$};
\node[youngnode] at(3.75,.75) (13) {$\yng(3,2)$};
\node[youngnode] at(4.5,1.5) (14) {$\yng(4,2)$};

\node[youngnode] at(4.5,0) (15) {$\yng(3,3)$};

\node[] at (0,0) (0) {};

\draw[->] (1) edge (2) (2) edge (3) (3) edge (4) (4) edge (5)(5) edge (6);
\draw[->] (7) edge (8) (8) edge (9)(9) edge (10) (10) edge (11);
\draw[->] (12) edge (13) (13) edge (14);
\draw[->] (1) edge (7);
\draw[->] (2) edge (8)(8) edge (12); 
\draw[->] (3) edge (9)(9) edge (13) (13) edge (15);
\draw[->] (4) edge (10)(10) edge (14);
\draw[->] (5) edge (11);
\draw[->] (6);
\end{tikzpicture}

\caption{The branching diagram. The highlighted path $\bm J = \left[\frac{1}{2} \rightarrow 0 \rightarrow \frac{1}{2} \rightarrow 1\right]$ corresponds to a set of five $4$-qubit quantum states: $\Ket{\bm J, M} = \Ket{M, J= 1, j_{[3]}=\frac{1}{2}, j_{[2]} = 0}$ with $M \in \lbrace -2, -1, 0, 1, 2 \rbrace$. The path takes the following sequence of steps: $\searrow,\nearrow,\nearrow$ and correponds to a Yamanouchi symbol $011$. Yamanouchi symbols are used in representation theory of the Symmetric group, which is made explicit by the diagram on the right, showing that each branching diagram node can be also labelled with the Young diagrams on two rows. Every Young diagram with $n$ boxes has $\frac{n}{2} + J$ boxes in the top and $\frac{n}{2}-J$ boxes in the bottom row labels a set of paths from $\mathcal{A}_n$ that end at the same $J$. As detailed in Appendix~\ref{App:PathsToYD}, the individual paths can be shown to be bijective with standard Young tableaux with two rows. We use to improve the sampling algorithm in Section \ref{Sec:ApproxSampl}.}
\label{Fig:Bratteli}
\end{figure*}
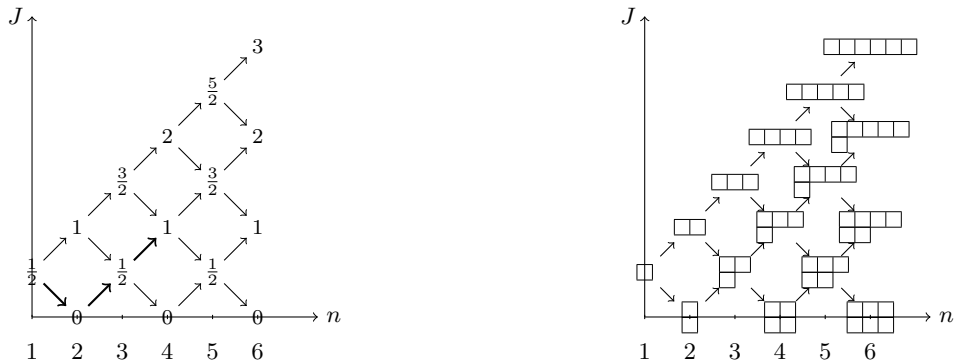

The studied circuits are derived from the Permutational Quantum Computing - a computational model based on recoupling of angular momenta \cite{Jordan09, Marzuoli05}. We hence review the basics of the angular momentum theory before introducing them. Consider $n$ qubits indexed by
$ [n] := \lbrace 1, 2 \ldots n \rbrace$.
The spin of the $k$-th qubit is defined by a triple of operators: %($\hbar = 1$):
\begin{align*}
\vec{S}_k &= \frac{1}{2} \left( X_k, \, Y_k, \,Z_k \right),
\end{align*}
where $X_k,Y_k,Z_k$ denote the Pauli $X,Y,Z$ operators on the $k$-th qubit. The \textit{total spin operator} on a qubit subset  $A \subseteq [n]$ is given by:
\begin{align*} S_A^2 &:= \sum_{k \in A} \vec{S}_k  \cdot  \sum_{k' \in A} \vec{S}_{k'} . \end{align*}
We write $S^2 := S^2_{[n]}$. The operators $S_A^2$ and $S_B^2$ commute if and only if the sets $A$ and $B$ are disjoint or one is contained in the other. %(Appendix~\ref{App:com1}). 
Let: \begin{align*} Z_A &:=  \frac{1}{2} \sum_{k \in A} Z_k, \end{align*} denote the \textit{azimuthal spin operator} on a qubit subset $A$. We again use $Z_{[n]} := Z$.  
The operators $Z_A$ and $S^2_A$ commute for any $A \subseteq [n]$ %(Appendix~\ref{App:com2}) 
and share an eigenspace labeled by quantum numbers $j_A$~and~$m_A$. The quantum number $j_A$ is the \textit{total spin} of qubits in $A$ and $m_A$ is the \textit{azimuthal spin}. Both spin numbers are subject to constraints: the azimuthal spin $m_A$ only takes values in integer steps between $-j_A $ and $j_A$, while the total spin numbers are either integer or half-integer and combine according to the angular momentum addition rules~\cite{Woit17,Sakurai94}: 
\begin{align}
j_{A \cup B}  \in \left\lbrace |j_A - j_B|,\, |j_A - j_B|+1,  \ldots ,\, j_A + j_B \right\rbrace. \label{Eq:coupling}
\end{align}

Sets of commuting spin operators can be used to define complete orthonormal bases \cite{Jordan09}. A particular basis is given by coupling a qubit at a time; that is by the joint eigenstates of: \begin{align*} 
 S_{ [2]}^2, \,S_{ [3] }^2, \, \ldots S^2, Z. \end{align*} We call it the \textit{sequentially coupled basis}. The basis states are labeled by eigenstates $j_{[2]}, j_{[3]} \ldots, j_{[n-1]}, J$ and $M$ of the spin operators. By Eq.~\ref{Eq:coupling}, these are subject to:
 \begin{align} j_{[1]} &= \frac{1}{2}, & j_{[k+1]} &=  \left|j_{[k]} \pm \frac{1}{2}\right| \label{Eq:couplingSeq}, \end{align}
which can be expressed diagrammatically by a \textit{branching diagram} (Fig.~\ref{Fig:Bratteli}). 
Up to the quantum number $M$, the sequential basis states correspond to paths in this diagram that start at $j_{[1]} = \frac{1}{2}$.

Let $\mathcal{A}_k$ be the set of all such paths on $k$ qubits. Any path $\bm j \in \mathcal{A}_k$ can be labelled by a bitstring by writing $1$ for any $\nearrow$ edge of the path and $0$ for an $\searrow$ edge of the path $\bm j$ in the branching diagram. For example:
\begin{align*}
\left[\frac{1}{2} \rightarrow 1 \rightarrow \frac{1}{2} \rightarrow 1 \right] \mapsto 101.
\end{align*}

Any prefix of length $m \leq k-1$  in such a bitstring contains at most  $\lceil \frac{m}{2} \rceil$ zeroes, since the path never goes below the horizontal axis of the branching diagram. 
These bitstrings play a role in the representation theory of the symmetric group and are called Yamanouchi symbols \cite{Pauncz67,Coleman68}. The sets of Yamanouchi symbols with the same Hamming weight correspond to Young diagrams on two rows, which can be seen in Fig.~\ref{Fig:Bratteli}. This is underpinned by the $\textsf{SU}(2)$ Schur-Weyl duality, that states that the $n$-qubit Hilbert space decomposes into the tensor product of the symmetric group $S_n$ modules (isomorphic to the Young diagrams on two rows) and the special unitary group $\textsf{SU}(2)$ under their joint action. 

 See Appendix~\ref{App:PathsToYD} for additional details of this correspondence and \cite{Kirby17b,Harrow05} for detailed discussion of the underlying representation theory.

\label{Subsection:QuantumSchurTransform}
For the sequentially coupled basis, the $\textsf{SU}(2)$ Schur-Weyl duality gives the $\textsf{SU}(2)$ Quantum Schur Transform as described in \cite{Bacon06,Bacon06b, Kirby17,Harrow05,Kirby17b,Krovi18}. It is a sequence of the Clebsch-Gordan transformations, that couple $j$ and $j'$ eigenspaces into a $\ket{J,M, j, j'}$ state by:
\begin{align*}
\ket{J,M, j, j'} &= \sum_{m, m'} C^{J,M}_{j,m ; \; j',m'} \ket{j,m} \ket{j',m'},
\end{align*}
where the summation over $m$ runs from $-j$ to $j$ in integer steps (and similarly for $m'$) and the $C^{J,M}_{j,m;\; j',m'}$ are the Clebsch-Gordan coefficients. The transform between the computational and the sequentially coupled basis is given by a cascade of the Clebsch-Gordan transforms \cite{Harrow05, Kirby17}. 
For example on $3$ qubits:
\begin{align*}
&\ket{J,M, j_{[2]}} \\
&= \sum_{m_1,m_2} \sum_{m_{[2]},m_3} C^{J,M}_{j_{[2]},m_{[2]}; \frac{1}{2}, m_3} C^{j_{[2]}, m_{[2]}}_{\frac{1}{2}, m_1 ;  \frac{1}{2}, m_2} \ket{m_1 m_2 m_3} \\
&= \sum_{m_1 m_2 m_3} \left[ {U_\textsf{Sch}} \right]^{J, M, j_{[2]}}_{m_1m_2m_3} \ket{m_1 m_2 m_3}.
\end{align*}
where we omitted the $j = \frac{1}{2}$ numbers for qubits for brevity.
The extension to the $n \geq 3$ qubit case is straightforward. We label the sequentially coupled basis states on $n$ qubits by $\ket{\bm J, M}$, where $\bm J$ is a path in $\mathcal{A}_n$.  

Permutational Quantum Computing in the sequentially coupled basis uses the permutation gate between two sequentially coupled basis states. Its transition amplitudes are:
\begin{align*}
\braket{\bm J,M | U_\pi|\bm J',M'},
\end{align*}
where the permutation gate $U_\pi$ is defined by its action on a computational basis state $\ket{x_1 \ldots x_n}$ as:
\begin{align*}
U_\pi \ket{x_1 x_2 x_3 \ldots x_n} &= \ket{x_{\pi(1)} x_{\pi(2)} x_{\pi(3)} \ldots x_{\pi(n)}}.
\end{align*}

\begin{figure}[t]
\begin{tikzpicture}[scale = 1]
	%%% BOTTOM TREE
    \node[circle, inner sep=2pt] at(-.65, 0.75) (1B) {$1$};
    \node[circle, inner sep=2pt] at(.65, 0.75) (3B) {$3$};
    \node[circle, inner sep=2pt] at(1.3, 0.75) (4B) {$4$};
    \node[circle, inner sep=2pt] at(1.95, 0.75) (5B) {$5$};
    \node[circle, inner sep=2pt] at(-0.0, 0.75) (2B) {$2$};

    \node at(-.55, -0.15) {$j_{[2]}$};
    \node at(-.15, -0.55) {$j_{[3]}$};
    \node at(.15, -1.1) {$j_{[4]}$};
    \node at(.6, -1.8) {$J, M$};
    
    \draw[color=gray,thick] (-0.65, 0.6) -- (-0.3, 0.1); 
    \draw[color=gray,thick] (-0, 0.6) -- (-0.3, 0.1); 
    \draw[color=gray,thick] (0.65, 0.6) -- (0.3, 0.1); 
    \draw[color=gray,thick] (1.3, 0.6) -- (0.3, -.7); 
    \draw[color=gray,thick] (1.95, 0.6) -- (0.6, -1.1); 
    \draw[color=gray,thick] (-0.3,0.1)  -- (0, -.3);
    \draw[color=gray,thick] (0,-0.3)  -- (0.3, -.7); %
    \draw[color=gray,thick] (0.3,-.7)  -- (0.6, -1.1);%
    \draw[color=gray,thick] (0.3,0.1) -- (0, -0.3);
    \draw[color=gray,thick] (0.6, -1.1) -- (0.6, -1.5);
    
    \node[circle, fill = black,inner sep=1pt] at (-0.3, 0.1){};
    \node[circle, fill = black,inner sep=1pt] at (0, -0.3){};
    \node[circle, fill = black,inner sep=1pt] at (0.3, -0.7){};
    \node[circle, fill = black,inner sep=1pt] at (0.6, -1.1){};

    %%% TOP TREE
    \node[circle, inner sep=2pt] at(-.65, 1.75) (1U) {$1$};
    \node[circle, inner sep=2pt] at(.65, 1.75) (3U) {$3$};
    \node[circle, inner sep=2pt] at(1.3, 1.75) (4U) {$4$};
    \node[circle, inner sep=2pt] at(1.95, 1.75) (5U) {$5$};
    \node[circle, inner sep=2pt] at(-0.0, 1.75) (2U) {$2$};
   
    \draw[color=gray,thick] (-0.65, 1.9) -- (-0.3, 2.4); 
    \draw[color=gray,thick] (0, 1.9) -- (-0.3, 2.4); 
    \draw[color=gray,thick] (0.65, 1.9) -- (0.3, 2.4); 
    \draw[color=gray,thick] (1.3, 1.9) -- (0.3, 3.4); 
    \draw[color=gray,thick] (1.95, 1.9) -- (0.6, 3.9); 
    \draw[color=gray,thick] (-0.3,2.4)  -- (0, 2.9);
    \draw[color=gray,thick] (0,2.9)  -- (0.3, 3.4); %
    \draw[color=gray,thick] (0.3,3.4)  -- (0.6, 3.9);%
    \draw[color=gray,thick] (0.3,2.4) -- (0, 2.9);
    \draw[color=gray,thick] (0.6, 3.9) -- (0.6, 4.6);
    
     \node[circle, fill = black,inner sep=1pt] at (-0.3, 2.4){};
    \node[circle, fill = black,inner sep=1pt] at (0, 2.9){};
    \node[circle, fill = black,inner sep=1pt] at (0.3, 3.4){};
    \node[circle, fill = black,inner sep=1pt] at (0.6, 3.9){};
    
    \node at(-.5, 2.8) {$j_{[2]}'$};
    \node at(-.15, 3.4) {$j_{[3]}'$};
    \node at(.16, 3.9) {$j_{[4]}'$};
    \node at(.6, 4.8) {$J', M'$};
    
    % PERMUTATION
    \draw (1B) -- (3U);
    \draw (2B) -- (1U);
    \draw (3B) -- (2U);
    \draw (4B) -- (5U);
    \draw (5B) -- (4U);

\end{tikzpicture}
\caption{Schematic representation of Permutational Quantum Computing in the sequentially coupled basis. The applied permutation is $(1,2,3)(4,5)$. } \label{Fig:PQC}
\end{figure}
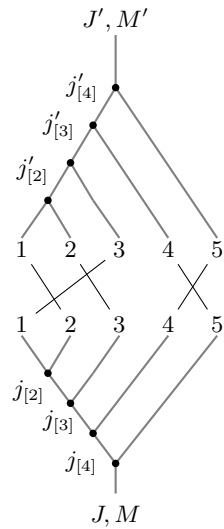

Both $Z$ and $S^2$ operators commute with $U_\pi$ and consequently, $M = M'$ and $J = J'$. 
The matrix $\braket{\bm J,M | U_\pi|\bm J',M'}$ block-diagonalizes to $J,M$ blocks; each of which corresponds the an irreducible representation of the symmetric group in the Young's orthogonal form. The transition amplitudes are then the matrix elements of these matrices \cite{Jordan09}.  Approximating them to polynomial precision was conjectured hard classically in \cite{Jordan09, Jordan08} but an efficient classical algorithm was found in \cite{Havlicek18}. 

The methods we present here work for a broader family of quantum circuits we call the $\textsf{SU}(2)$ \textit{Quantum Schur Sampling} circuits. These have transition amplitudes: \begin{align*}
\braket{\bm J,M |W|\bm J',M'} ,
\end{align*}
where $W$ is defined by its action on a computational basis state $\ket{x}, x \in \lbrace 0,1 \rbrace ^n $:
\begin{align*}
W\ket{x} &= \ket{w(x)},
\end{align*}
with $w: \lbrace 0, 1\rbrace^n \rightarrow \lbrace 0,1 \rbrace^n$ being a classical function given by a sequence of Toffoli gates -- we consider only such $W$ where this sequence is $\poly(n)$ long \footnote{These circuits extend the notion of quantum Schur sampling circuits introduced in \cite{Havlicek18}, where we only considered circuits of the form $\braket{\bm J,M |\Lambda|\bm J',M'}$ with $\Lambda$ being a $Z$-diagonal gate with efficiently computable elements. Our technique works for these circuits as well.}.

The circuits become similar in structure to Shor's algorithm in a sense of Fig.~\ref{Fig:Circuits} if we allow for ancillary qubits. The simulation results apply also to these circuits, which we discuss in Section~\ref{Section:Ancillas}.

\section{Finding Large Probabilities}
We now describe an algorithm for finding large probabilities in the output of the circuits (see Fig.~\ref{Fig:PQC-SEQ}).
Our approach is built on the concept of computational tractability introduced in \cite{VanDenNest09}: 

\begin{definition}\label{def:tractability}
An $n$-qubit state $\ket{\psi}$ is computationally tractable (CT) if it is possible to classically sample from the distribution:
\begin{align*} p(x) &= \lbrace |\braket{x | \phi} |^2: x \in \lbrace 0,1\rbrace^n \rbrace, \end{align*}  in polynomial time and
the overlaps $\braket{x | \phi}$ can be computed to exponential precision for a computational basis state $\ket{x}$ in polynomial time.
\end{definition}

We proved in \cite{Havlicek18} that the sequentially coupled basis states are CT. As a corollary, we show that $\ket{\phi} = W \ket{\bm J,M}$ is also CT: 
\begin{lemma}
 $\ket{\phi} = W\ket{\bm J, M}$ is CT. 
\label{Lemma:CTExtension}
\end{lemma}
\begin{proof}
Since $\braket{x | \bm J, M}$ can be efficiently computed because $\ket{ \bm J, M}$ is CT, so can be $\braket{x | W | \bm J, M} = \braket{w^{-1}(x) | \bm J, M} $. The distribution:
\begin{align*} p(x) &=  |\braket{x | W| \bm J, M} |^2, \end{align*}
can be efficiently sampled by applying the inverse of $w(x)$ to the samples drawn from $|\braket{x| \bm J, M}|^2$:
\begin{align*} p(w^{-1}(x)) &= |\braket{w^{-1}(x) | \bm J, M} |^2 = |\braket{x | W| \bm J, M} |^2.\end{align*}
Since $W$ is made of polynomially-many Toffoli gates, the inverse is obtained by applying the circuit in reverse to the bitstring $x$. \end{proof}
We also state Lemma 3 of \cite{VanDenNest09}, which is an application of the Chernoff-Hoeffding bound.
\begin{lemma}[CT state overlap $(\epsilon, \delta)$-approximation \cite{VanDenNest09}]
An overlap $\braket{\phi | \psi}$ between two CT states can be approximated by $\tilde{a}$, such that:
\begin{align*}
| \tilde{a} - \braket{\phi | \psi}| \leq \epsilon,
\end{align*}
with probability $1-\delta$ 
in $\poly(\frac{1}{\epsilon}, n, \log \frac{1}{\delta})$ time. We say that the overlap $\braket{\phi | \psi}$ is $(\epsilon, \delta)$-approximated by $\tilde{a}$.
\label{Lemma:CTOverlaps}
\end{lemma}
 
 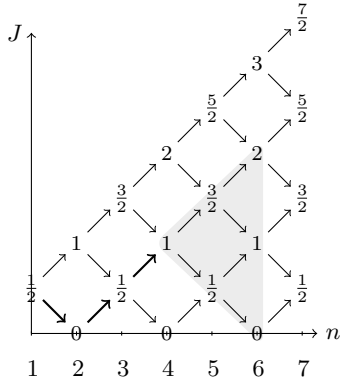
\begin{figure}[t]

\begin{tikzpicture}[scale=.8]
\fill[rounded corners,fill=gray!15] (2.8,1.5)--(4.6,-.2)--(4.6,3.2)--cycle; 
\draw [->](.75,0)--(.75,5) node[left]{$J$};
\draw [->](.75,0)--(5.5,0) node[right]{$n$};
\foreach \x/\xtext in {0.5/1, 1/2, 1.5/3, 2/4, 2.5/5, 3/6, 3.5/7}
{\draw (1.5*\x cm,1pt ) -- (1.5*\x cm,-1pt ) node[anchor=north] at(.75 + 1.5*\x cm, -0.3) {$\xtext$};}

\node[treenode] at(0.75,0.75) (1) {$\frac{1}{2}$};
\node[treenode] at(1.5,1.5) (2) {$1$};
\node[treenode] at(2.25,2.25) (3) {$\frac{3}{2}$};
\node[treenode] at(3,3) (4) {$2$};
\node[treenode] at(3.75,3.75) (5) {$\frac{5}{2}$};
\node[treenode] at(4.5,4.5) (6) {$3$};
\node[treenode] at(5.25,5.25) (16) {$\frac{7}{2}$};

\node[treenode] at(1.5,0) (7) {$0$};
\node[treenode] at(2.25,0.75) (8) {$\frac{1}{2}$};
\node[treenode] at(3,1.5) (9) {$1$};
\node[treenode] at(3.75,2.25) (10) {$\frac{3}{2}$};
\node[treenode] at(4.5,3) (11) {$2$};
\node[treenode] at(5.25, 3.75) (17) {$\frac{5}{2}$};

\node[treenode] at(3,0) (12) {$0$};
\node[treenode] at(3.75,.75) (13) {$\frac{1}{2}$};
\node[treenode] at(4.5,1.5) (14) {$1$};
\node[treenode] at (5.25, 2.25) (18) {$\frac{3}{2}$};

\node[treenode] at(4.5,0) (15) {$0$};
\node[treenode] at (5.25, 0.75) (19) {$\frac{1}{2}$};

\node[] at (0,0) (0) {};

\draw[->] (1) edge (2) (2) edge (3) (3) edge (4) (4) edge (5)(5) edge (6)(6) edge (16);
\draw[->,thick] (7) edge (8) (8) edge (9); 
\draw[->] (9) edge (10) (10) edge (11) (11) edge (17);
\draw[->] (12) edge (13) (13) edge (14) (14) edge (18);
\draw[->] (15) edge (19);
\draw[->,thick] (1) edge (7);
\draw[->] (2) edge (8)(8) edge (12); 
\draw[->] (3) edge (9)(9) edge (13) (13) edge (15);
\draw[->] (4) edge (10)(10) edge (14) (14) edge (19);
\draw[->] (5) edge (11) (11) edge (18);
\draw[->] (6) edge (17);
\end{tikzpicture}

\caption{The summation $\sum_{\bm J \supseteq \bm j}$ runs over all paths $\bm J \in \mathcal{A}_n$ that contain $\bm j \in \mathcal{A}_k$.  As an example, in the diagram above $\bm j = \left[ \frac{1}{2} \rightarrow 0 \rightarrow \frac{1}{2} \rightarrow 1 \right]$ and $k = 4, n = 6$. The summation runs over the paths within the shaded region. It follows that (in terms of the Yamanouchi symbols) $\bm J \supseteq \bm j = \lbrace{01111,01110,01101,01100 \rbrace}$.}
\label{Fig:SumOverFutures}
\end{figure}
 
 We now show how to approximate a set of output probability marginals, an enabling result for extension of the techniques used by Schwarz and van den Nest in \cite{Schwarz13}. Given a path $\bm j \in \mathcal{A}_k$ for $k \leq n$, define the \textit{output marginal} $p (\bm j)$:
\begin{align*}
p \left(\bm j \right) &:= \sum_{\bm J \supseteq \bm j} \sum_M p \left( \bm J, M\right) \\ 
&= \bra{\phi}  \sum_{\bm J \supseteq \bm j;M}   \Ket{\bm J, M} \Braket{\bm J, M| \phi} := \bra{ \phi} \Pi \left(\bm j\right) \ket{\phi},
\end{align*} 
where the summation $\sum_{\bm J \supseteq \bm j}$ sums all paths $\bm J \in \mathcal{A}_n$ that contain  $\bm j \in \mathcal{A}_k$ (see Fig.~\ref{Fig:SumOverFutures}). The summation $\sum_M$ runs from $-J$ to $J$ in integer steps. We use $\sum_ {\bm J \supseteq \bm j; M}$ as a shorthand for $ \sum_{\bm J \supseteq \bm j} \sum_M$. The projector:
\begin{align*}
\Pi \left(\bm j \right) &:= \sum_{\bm J \supseteq \bm j; M} \Ket{\bm J, M} \Bra{\bm J, M},
\label{Eq:MarginalProjector}
\end{align*}
can be simplified to (Appendix~\ref{App:Projector}):
\begin{align*}
\Pi \left({\bm j}\right) &=
\sum_{m} \Ket{ \bm j, m} \Bra{\bm j, m}.
\end{align*}
where the sum $\sum_m$ runs over $m \in \lbrace -j, -j+1,\ldots j \rbrace$.
\begin{lemma}
For $\bm j \in \mathcal{A}_k$, 
$p(\bm j)$ can be classically $(\epsilon, \delta)$-approximated by
 $\tilde{p}(\bm j)$ in $\poly \left(\frac{1}{\epsilon}, n, \log \frac{1}{\delta} \right)$ time.
\label{Lemma:marginalappx}
\end{lemma}
\begin{proof}
 We first show that the marginal $p(\bm j)$ on $k$ qubits can be written as a transition amplitude of a larger, $(n+k)$-qubit circuit as: 
\begin{align*}
\braket{\phi|\bm j, m}  \braket{\bm j, m | \phi} 
&= \left( \bra{\bm j, m} \bra{\phi} \right) U_{\textsf{SWAPS}} \left( \ket{\phi} \ket{\bm j,m} \right), 
\end{align*}
where $U_{\textsf{SWAPS}}$ is a permutation gate on $k+n$ qubits. Write symbolically $\ket{\bm j, m} = \ket{\psi} = \ket{\psi_1\psi_2\ldots\psi_k}$ and $\ket{\phi} = \ket{\phi_1\phi_2\ldots\phi_n}$, so that:
\begin{align*}
\braket{\phi|\bm j, m} &\braket{\bm j, m|\phi} = \braket{\phi|\psi} \braket{\psi|\phi}\\&= \braket{\phi_1 \ldots \phi_n |\psi_1 \ldots \psi_k }\braket{\psi_1 \ldots \psi_k |\phi_1 \ldots \phi_n }.
\end{align*}
Let $U_\textsf{SWAPS}$ swap the $(n+i)$-th and $(n-k+i)$-th qubits for all $1 \leq i  \leq k$:
\begin{align*}
U_\textsf{SWAPS} &\ket{\psi_1 \ldots \psi_k }\ket{\phi_1 \ldots \phi_n}\\&= \ket{\phi_1 \ldots \phi_k }\ket{\psi_1 \ldots \psi_k, \phi_{k+1} \ldots \phi_n}.
\end{align*}
This gives:
\begin{align*}
& \bra{\phi} \bra{\psi} U_\textsf{SWAPS} \ket{\psi}\ket{\phi} \\
&= \bra{\phi_1 \ldots \phi_n} \braket{\psi_1 \ldots \psi_k | \phi_1 \ldots \phi_k }\ket{\psi_1 \ldots \psi_k, \phi_{k+1} \ldots \phi_n} \\
&= \braket{\phi_1 \ldots \phi_n | \psi_1 \ldots \psi_k}  \braket{\psi_1 \ldots \psi_k | \phi_1 \ldots \phi_n} \\
&= \braket{\phi|\psi} \braket{\psi|\phi},
\end{align*}
as desired. Since both $\ket{\bm j, m}$ and $\ket{\phi}$ states are CT and $U_\textsf{SWAP}$ is a permutation on up to $2n$ objects,  \begin{align*}
\left( \bra{\bm j, m} \bra{\phi} \right) U_\textsf{SWAPS} \left( \ket{\phi} \ket{\bm j,m} \right),\end{align*}
 can be $(\epsilon, \delta)$-approximated by Lemma~\ref{Lemma:CTOverlaps}. Therefore: 
  \begin{align*}
    p(\bm j)  &=  \bra{\phi} \Pi(\bm j) \ket{\phi} \\
    &= \sum_m \braket{\phi|\bm j, m}  \braket{\bm j, m | \phi} 
     \\ &= \sum_{m} \left( \bra{\bm j, m} \bra{\phi} \right) U_{\textsf{SWAPS}} \left( \ket{\phi} \ket{\bm j,m} \right).
\end{align*}
 
Since $\sum_m$ sums $2j+1 \leq n+1$ terms, it follows that $p(\bm j)$ can be also $(\epsilon, \delta)$-approximated. 
\end{proof}

We now combine Lemmas~\ref{Lemma:CTExtension},~\ref{Lemma:CTOverlaps} and~\ref{Lemma:marginalappx} to describe a classical algorithm that finds large elements in the output distribution of quantum Schur circuits. It is an adaptation of the Kushilevitz-Mansour algorithm \cite{kushilevitz1993learning}.

%%% KUSHILEVITZ-MANSOUR
\begin{theorem}
\label{Theorem:Kushilevitz-Mansour}
Let $p(\bm J): \mathcal{A}_n \rightarrow [0,1]$ be a probability distribution on paths. There is a classical algorithm that outputs a set $L \subseteq \mathcal{A}_n$ in $\poly \left(n,\frac{1}{\theta}, \log \frac{1}{\gamma} \right)$ time, such that for some $\theta > 0$:
\begin{equation}
\begin{aligned}
\label{Eq:LConditions}
\forall \, \bm J \in L: \; p(\bm J) &\geq  \frac{\theta}{2}, \\
\forall \, \bm J \in \mathcal{A}_n: \; p(\bm J) &> \theta \implies \bm J \in L,
\end{aligned}
\end{equation}
 with probability at least $1-\gamma$. 
\end{theorem}
\begin{proof} 
 See Algorithm~\ref{Algorith:KM}.
 
%%% KUSHILEVITZ MANSOUR ALGORITHM
\RestyleAlgo{boxruled}
\begin{algorithm}[h]
\caption{Kushilevitz-Mansour Algorithm}
\label{Algorith:KM}
\begin{enumerate}
\item Set $L_2 = \emptyset$. Choose:
\begin{align*}
\delta < \frac{\theta}{2n},\end{align*} and compute $\tilde{p}(\bm j_2)$ for both paths in $\mathcal{A}_2$ by Lemma~\ref{Lemma:marginalappx}
, such that:
\begin{align*} |p(\bm j_2) - \tilde{p}(\bm j_2)| \leq \frac{\theta}{4}.\end{align*} 
Add $\bm j$ to $L_2$ if $\tilde{p}(\bm j_2) \geq \frac{3}{4}\theta$.

\item Continue for $k = 3, \ldots, n$. Assume $L_{k-1}$ 
has been found. 
For any path $\bm j_{k-1} \in L_{k-1}$, take all possible steps in the branching diagram. This gives paths $\bm j_k \in \mathcal{A}_k$ 
that end at $j_k = |j_{k-1} \pm \frac{1}{2}|$. 

For every such path, compute the approximation $\tilde{p}(\bm j_k)$ such that:
\begin{align*}
\left| \tilde{p}(\bm j_k) -p(\bm j_k) \right| \leq \frac{\theta}{4}. 
\end{align*}
 If $\tilde{p} (\bm j_k) \geq \frac{3}{4}\theta$, add the path $\bm j_k$ to $L_k$. 

\item In every step of the computation, check if $|L_k|>\frac{2}{\theta}$. If true, halt and output $\emptyset$. 

\textit{The algorithm never halts if all approximation steps succeed.}

\item Output $L = L_n$. 
\end{enumerate}
\end{algorithm}

The algorithm runs in $n$ steps, each of which succeeds with probability at least $(1-\delta)^{|L_k|}$. Since $|L_k| \leq \frac{2}{\theta}$,  the success probability is at least: 
\begin{align*} (1-\delta)^{2n/\theta} \geq 1- \frac{2 \delta n}{\theta} := 1 - \gamma.
\end{align*} Thanks to $\delta < \frac{\theta}{2n}$, it follows that $1-\gamma>0$. The algorithms terminates in $\poly \left(n, \frac{1}{\theta}, \log \frac{1}{\gamma} \right)$ time as it halts whenever the number of elements in any list exceeds $\frac{2}{\theta}$. Since $p(\bm j_k) \geq \frac{\theta}{2}$ for each $\bm j_k \in L_k$,  the final list $L$ contains at most $\frac{2}{\theta}$ elements by normalization. So if all approximation steps succeed, the algorithm does not halt before it outputs $L$. 
\end{proof}

Algorithm~\ref{Algorith:KM} has an interesting consequence: since it runs in polynomial time whenever $\theta = 1/\poly(n)$, paths with polynomially small $p(\bm J) = \sum_M  p(\bm J, M)$ can be found in polynomial time.  When such path $\bm J$ is found, it is possible to approximate $p(\bm J, M)$ for all $M$, since there are  $2J +1 \leq n+1$ distinct values of $M$ it has to be approximated for by Lemma \ref{Lemma:CTOverlaps}. Such approximation of transition amplitudes has the same precision as if when polynomially many samples were taken with a quantum computer. The $\textsf{SU}(2)$ Quantum Schur sampling circuits therefore cannot encode classically hard-to-approximate quantities in amplitudes that could be resolved by sampling, because any such quantity could be found by the presented algorithm and then approximated by Lemma~\ref{Lemma:CTOverlaps}.

\section{Approximate Sampling}
\label{Sec:ApproxSampl}
Following Schwarz and van den Nest \cite{Schwarz13}, we use the above algorithm to approximately sample the quantum Schur circuits under additional sparsity constraint on their output distribution: 

%%% SCHWARZ-VAN DEN NEST ALGORITHM
\begin{definition}[$\epsilon$-approximate $t$-sparsity]
A probability distribution $p(\bm J, M)$ is $t$-sparse if it has at most $t$ non-zero elements $p(\bm J, M)$. A probability distribution $\tilde{p}(\bm J, M)$ is $\epsilon$-approximately $t$-sparse if there exists a $t$-sparse distribution $p(\bm J, M)$ such that:
\begin{align*}
\| p - \tilde{p} \|_1 \leq \epsilon.
\end{align*}
\end{definition}
We also adapt a technical lemma from \cite{kushilevitz1993learning}.
\begin{lemma}
\label{Lemma:SupportLemma}
Let $p(\bm J, M)$ be an $\epsilon$-approximate $t$-sparse distribution. Let $S$ be the set of all $(\bm J, M)$ for which $p(\bm J, M)$ is greater than $\frac{\epsilon}{t}$. Then:
\begin{align*}
\sum_{\bm J \not\in S; M}  p(\bm J, M) \leq 2 \epsilon.
\end{align*}
\end{lemma}

\begin{proof}
Let $p^{t}(\bm J, M) $ be a $t$-sparse distribution that is  $\epsilon$-close to $p(\bm J, M)$. Define $T$ to be the set of all $(\bm J, M)$ for which $p^t$ is nonzero, i.e. the support of $p^{t}$.
Trivially, $S \cap T \subseteq S$, which implies that: 
\begin{align*}
\sum_{\bm J \not\in S; M}  p(\bm J, M) &\leq \sum_{\bm J \not\in S\cap T; M}  p(\bm J, M).
\end{align*}
Define the indicator $I_A: A \rightarrow \lbrace 0, 1 \rbrace$ on the set $A$ as follows:
\begin{align*}
I_A(a) &:= \begin{cases} 1, \text{ if } a \in A, \\ 
0, \text{ otherwise}, \end{cases} 
\end{align*}
so that: 
\begin{align*}
\sum_{\bm J \not\in S\cap T; M}  p(\bm J, M) &= \sum_{\bm J \in \mathcal{A}_n; M}p(\bm J,M)  \left( I_{S \cap T}(\bm J, M) - 1 \right) \\&= \| p I_{S \cap T} - p \|_1.
\end{align*}
By the triangle inequality: 
\begin{align*}
 \| p I_{S \cap T} - p \|_1
&\leq \| p - p I_{T} \|_1 + \| pI_{S \cap T} - p I_{T} \|_1.
\end{align*}
Since $p$ is $\epsilon$-approximate $t$-sparse, it follows that:
\begin{align*}
&\| p - pI_T \|_1 \leq \| p - p^{t} \|_1  \leq \epsilon. \label{Eq:Lemma4A}
\end{align*}
We also have that: 
\begin{equation}
\begin{aligned}
\| p I_{S \cap T} - p I_{T} \|_1 &= \sum_{\bm J \in T;M} p(\bm J,M)  I_{S \cap T}(\bm J,M) \\
&= \sum_{\bm J \in T/S; M} p(\bm J,M) \leq \frac{\epsilon}{t}t = \epsilon.
\label{Eq:Lemma4B}
\end{aligned}
\end{equation}
because all elements in $T/S$ are $\leq \frac{\epsilon}{t}$ and there is at most $t$ of them.
This gives: 
\begin{align*}
\sum_{\bm J \not\in S; M}  p(\bm J, M)  \leq 2\epsilon.
\end{align*}
\end{proof}

Theorem~\ref{Theorem:Kushilevitz-Mansour} and Lemma~\ref{Lemma:SupportLemma} can be combined to define a probability distribution close to the quantum output that can be sampled from in $\poly(t, \frac{1}{\epsilon}, n)$ time. Assume that $p(\bm J, M)$ is $\epsilon$-approximately $t$-sparse. Let $L \subseteq \mathcal{A}_n$ be the set of paths generated by the Kushilevitz-Mansour algorithm with threshold $\theta = \frac{\epsilon}{t}$. Choose: 
\begin{align} \epsilon' &= \min \left( \frac{\epsilon}{(n+1)|L|}, \, \frac{\epsilon}{4t} \right), 
\label{Eq:EpsilonDash}
\end{align} and compute $\epsilon'$ approximations $\tilde{p}(\bm J, M)$ for all $\bm J \in L$ and $M$ by Lemma~\ref{Lemma:CTOverlaps}. Define a normalization factor $\alpha$ as:  \begin{align*}
   \alpha &= \frac{1}{2^n- \sum_{\bm J\in L}(2J + 1)},
\end{align*}
such that $\sum_{\bm J \not\in L; M} \alpha = 1$  (see Appendix \ref{App:completeness}). Use the $\epsilon'$-approximations $\tilde{p}(\bm J, M)$ to define: \begin{align*}
\tilde{p} &= \begin{cases} \tilde{p} (\bm J,M) \, \text{ for } \bm J \in L, \\
 \tilde{p}_\circ  := \alpha(1- \sum_{\bm J \in L; M} \tilde{p} (\bm J, M)) \text{ otherwise, }\end{cases},
\end{align*}
so that $\tilde{p}$ becomes uniform on all $\bm J$ outside $L$. The constant $ \tilde{p}_\circ $ is chosen so that $\tilde{p}$ is normalized. 
Then:
 {\small
\begin{align*}
&\| \tilde{p} - p\|_1 = \\ & \sum_{\bm J \in L; M} | \tilde{p}(\bm J, M) - p (\bm J, M)| + \sum_{\bm J \not\in L; M} | \tilde{p}(\bm J, M) - p(\bm J, M)| \\ 
&\leq \epsilon + \sum_{\bm j \not\in L;M} | \tilde{p}(\bm J, M) - p(\bm J, M)|, 
\end{align*}}
since:
\begin{align*}
\sum_{\bm J \in L; M} | \tilde{p}(\bm J, M) - p (\bm J, M)| &\leq  \sum_{\bm J \in L; M} \epsilon' \\ &\leq (n+1)|L| \epsilon' \leq \epsilon.
\end{align*}
Define also:
\begin{align*}
 p_\circ & = \alpha \left(1 - \sum_{\bm J \in L; M} p(\bm J, M)\right),
\end{align*}
and notice that:
\begin{align*}
\sum_{\bm J \not\in L; M}|  p_\circ  -  \tilde{p}_\circ |&\leq \left|  \sum_{\bm J \in L; M} p(\bm J, M) - \sum_{\bm J \in L; M} \tilde{p}(\bm J, M) \right| \\&\leq 
\sum_{\bm J \in L; M} |p(\bm J, M) - \tilde{p}(\bm J, M) | \leq \epsilon.
\end{align*}
By the triangle inequality:
\begin{align*}
&\sum_{\bm J \not\in L;M} | \tilde{p}(\bm J, M) - p(\bm J, M)| = \sum_{\bm J \not\in L;M} |  \tilde{p}_\circ- p (\bm J,M)| \\
&\leq \sum_{\bm J \not\in L; M}|  p_\circ  -  \tilde{p}_\circ |   + \sum_{\bm J \not\in L; M} | p_\circ  - p (\bm J, M)| \\
&\leq \epsilon + \sum_{\bm J \not\in L; M} | p_\circ  - p (\bm J, M)|.
\end{align*}
We now use the set $S$ from from Lemma~\ref{Lemma:SupportLemma}. $S \subseteq L$ by the defining property of $L$. It follows that:
\begin{align*}
\sum_{\bm J \not\in L;M} p (\bm J, M) \leq \sum_{\bm J \not\in S;M} p (\bm J, M) \leq 2\epsilon.
\end{align*}
Notice that:
\begin{align*}
\sum_{\bm J \not\in L; M}  p_\circ  &= \sum_{\bm J \not\in L; M} \left( \alpha \sum_{\bm J' \not\in L; M'} p(\bm J', M') \right) \\ &= \sum_{\bm J' \not\in L; M'} p(\bm J', M'). \end{align*}
This gives:
\begin{align*}
\sum_{\bm J \not\in L; M} | p_\circ  - p (\bm J, M)| &\leq \sum_{\bm J \not\in L; M}  p_\circ  + \sum_{\bm J \not\in L; M} p (\bm J, M) \\ & = 2 \sum_{\bm J \not \in L ; M} p(\bm J,M) \leq 4 \epsilon, 
\end{align*}
which leads to:
\begin{align*}
\sum_{\bm J \not\in L; M}  | \tilde{p}(\bm J, M) - p (\bm J, M)|\leq 5\epsilon,
\end{align*}
and:
\begin{align*} \| \tilde{p} - p\|_1 \leq 6\epsilon.\end{align*}
We now show how to classically sample $\tilde{p}$.
\begin{theorem}
\label{Theorem:WeakSampling} Assume that $p(\bm J, M)$ is $\epsilon$-approximate $t$-sparse. It can be sampled classically in $\poly(n, \frac{1}{\epsilon}, t)$ time to $6\epsilon$ error in the total variational distance.
\end{theorem}
\begin{proof}
Use the Kushilevitz-Mansour algorithm in Theorem~\ref{Theorem:Kushilevitz-Mansour} with threshold $\theta = \frac{\epsilon}{t}$ to find $L$. Compute $b = \sum_{\bm J \in L; M} \tilde{p}(\bm J, M)$. Flip a coin with a bias $b$.
\begin{itemize}
\item With probability $b$, output a sample drawn from $\tilde{p}(\bm J, M)/b$ for $\bm J \in L$ and corresponding $M$.
\item With probability $1-b$, output $(\bm J, M)$ for $\bm J \not\in L$ uniformly randomly.
\end{itemize}
To sample $(\bm J,M)$ uniformly randomly,  generate a random bitstring with $n -1$ bits and check if it encodes a Yamanouchi symbol. This can be verified by checking that any prefix of $m \leq n-1$ bits has at most $\lceil \frac{m}{2} \rceil$ zeroes. Once found, generate a random integer $M'$ from $[n+1]$. Check if $M' \leq 2J+1$. If yes, define $M = (M' - J-1)$ and output $(\bm J, M)$. Otherwise repeat. A valid Yamanouchi symbol will be found in $\poly(n)$ trials by a dimensionality argument (Appendix~\ref{App:completeness}). This procedure samples the probability distribution $\tilde{p}$ defined above, which has been shown $6\epsilon$ close in the total variational distance to $p$.
\end{proof}

While the above algorithm  runs in $\tilde{p}$ in $\poly(\frac{1}{\epsilon}, n ,t)$ time, it discards significant amount of paths during the uniform sampling of paths which may be an unnecessary bottleneck for the eventual implementation. We explain how to avoid this problem by an alternative algorithm for sampling the paths, based on the Greene-Nijenhuis-Wilf algorithm \cite{Greene79} in Appendix~\ref{App:PathsToYD}. 

\section{How sparse is the output?}

\begin{figure}[t]
\includegraphics[scale=0.5]{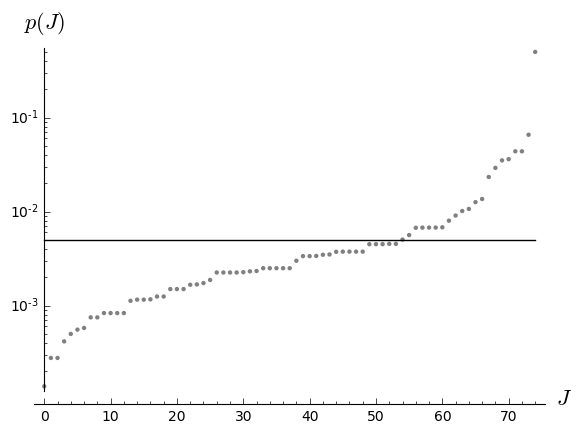}
\caption{Output distribution of PQC-SEQ that does not satisfy the sufficient sparsity condition for $n=10$.  The horizontal line labels the $(2 n^2)^{-1}$ threshold. The distribution is actually $1/10$-approximate $21$-sparse and `fools' the proxy criteria by having a single overwhelmingly large element. The $p$-axis is logarithmic.}
\label{Fig:FalsePositive}
\end{figure}

We consider the range of applicability of the outlined algorithm. Since the set of classical gates $W$ is large, we limit this analysis to Permutational Quantum Computing in the sequentially coupled basis and study the output distributions for $n \le 10$ qubits. We randomly chose $5$ paths and consider $10$ random permutations for each. This gives $50$ sets of output distributions with dimension $d$ determined by $J$ of each path. Recall that $d$ can be exponentially large in $n$. 

 All chosen distributions contained an element greater than $\frac{1}{2n}$. As a sufficient condition for $\frac{1}{n}$-approximate $2n^2$-sparsity  by Lemma~\ref{Lemma:SupportLemma}, we checked if the sum of all elements less than $\frac{1}{2n^2}$ is less than $\frac{1}{2n}$. Distributions for permutations on $4$ to $9$ qubits all have this property, while the fraction that do not have it for $n = 10$ qubits was estimated to be less than $0.1 \%$. Being a sufficient condition, some of these distributions are nevertheless very far from flat - an example is shown in Fig.~\ref{Fig:FalsePositive}.  
 
 We also consider a stricter sufficient condition: for all $J$-blocks with dimension $d > n$, we computed the fraction of  output distributions for which the sum of all elements except for the largest $C \left(  \log_2 d \right)^D$ ones is less than $1/ \log_2 d $ for some constants $C$ and $D$. Since $d < 2^n$, this condition suffices for $2n$-approximate $(C n^D)$-sparsity of the output. Almost all of the distributions, with the exception of about $0.4 \%$ of those for $n=9$, were $2 \log(d)$-approximate $ \log(d)^2$-sparse.
While we were not able to prove that a significant fraction of the output distributions are $\epsilon$-approximate $t$-sparse for some $t = \poly(n)$ and $\epsilon = 1/\poly(n)$, the results give some indication that close-to-sparse output distributions could be common for the relevant regime of Permutational Quantum Computing. 

\section{Circuits with Ancillas}
\label{Section:Ancillas}
The proposed simulation technique extends to quantum Schur sampling circuits with ancilla qubits, with transition amplitudes given by:
\begin{align*}
\left(\bra{0}^{k'} \bra{\bm J', M'} \right) W \left(\ket{\bm J, M} \ket{0}^{k}\right),
\end{align*}
for $\bm J' \in \mathcal{A}_{n'}$ and $\bm J \in \mathcal{A}_n$, such that $k'+n' = k+n$. Note that $W \left(\ket{\bm J, M} \ket{0}^{k}\right)$ is CT. 
Since the marginal approximation of Lemma~\ref{Lemma:marginalappx} relies only on approximating overlaps of the form:
\begin{align*}
\left( \bra{\bm j, m} \bra{\phi} \right) U_\textsf{SWAPS} \left( \ket{\phi} \ket{\bm j,m} \right),\end{align*}
where $\ket{\phi}$ is a computationally tractable state, it also extends to marginals:
\begin{align*}
\left( \bra{0}^{k'} \bra{\bm j, m} \bra{\phi} \right) U_\textsf{SWAPS} \left( \ket{\phi} \ket{\bm j,m} \ket{0}^{k'} \right).\end{align*}
since $\ket{\phi} \ket{\bm j,m} \ket{0}^{k'}$ is CT (see~\cite{VanDenNest09} for details). 

We give some evidence that these circuits can give rise to computationally interesting structures, largely inspired by \cite{Jordan08}. Prepare:
\begin{align*}
U_\textsf{Sch} H^{\otimes n} \ket{0}^{\otimes n} = \frac{1}{\sqrt{2^n}} \sum_{\bm J, M} \ket{\bm J, M},
\end{align*}
and consider a classical circuit $W$ that encodes the path $\bm J$ a the Yamanouchi symbol $x$ into an ancilla register. This should done \textit{before} the Schur transform as the $W$ gate is generally controlled in the computational basis. A way to implement this is to use the form of QST which encodes the information about irreps explicitly into the computational basis input at the expense of logarithmic overhead in number of qubits \cite{Harrow05, Kirby17b}. Additionally, compute the value of $J$ to another ancilla register of $\lceil\log n -1\rceil$ qubits, giving the state:
\begin{align*}
 \frac{1}{\sqrt{2^n}} \sum_{\bm J, M} \ket{\bm J, M} \ket{x(\bm J)}\ket{J}.
\end{align*}
Apply the permutation gate $U_\pi$ to the first register. After applying the gate sequence $H^{\otimes n} U_\textsf{Sch}^\dag$ and measuring the first $n$ qubits and the $J$ register, we have that:
\begin{align*} p(\underbrace{0 \ldots 0}_n, J) &= \frac{1}{4^n} \left|\sum_{\bm J} \braket{\bm J| U_\pi | \bm J} \right|^2. \end{align*}
where $\sum_{\bm J}$ runs over all paths that end at $J$. Here $T(\pi) = \sum_{\bm J}\bra{\bm J} U_\pi \ket{\bm J}$ is the trace of $U_\pi$ over the $J$-block, which is (up to a sign) the square of the character of the conjugacy class of $\pi$ of the $S_n$ irrep. defined by $J$. This quantity is known to be $\#\P$-hard by \cite{Hepler94}, so we know that there exist $\pi \in S_n$ for which \textit{exact} computation of $T(\pi)$ becomes intractable under the standard complexity theoretic assumptions. Despite the fact that an efficient classical method for computing \textit{additive} approximations to this quantity was given by \cite{Jordan08} (its existence is in fact a consequence of Theorem~\ref{Theorem:Kushilevitz-Mansour}), it is still possible that its multiplicative approximation retains hardness. This could lead to another class of probability distributions unlikely to be sampled from classically akin to \cite{Bremner11, Bremner16, Aaronson10}. 
 On the discouraging side, limitations of the quantum `Fourier-Schur' sampling in the context of addressing the hidden subgroup problem were identified in \cite{Childs06}.

\section{Discussion}
Circuits using the QST underpin a diverse range of protocols in quantum information processing, from state discrimination to computational models such as Permutational Quantum Computing. While studying the computational power of the transform, we singled out a class of circuits with QST blocks that extend a computationally interesting regime of Permutational Quantum Computing. 
The key result that enabled this analysis was the efficient approximation of quantum Schur sampling circuits studied in~\cite{Havlicek18} as means to characterize its computational power. Building on the work of Schwarz and Van den Nest~\cite{VanDenNest09,Schwarz13}, we showed that large elements of the output distributions can be efficiently found, which precludes the possibility that the circuits could encode quantities that would be hard to classically approximate by taking polynomial number of samples. 
 	We subsequently proved that these circuits can be classically efficiently approximately sampled from if their output distribution becomes sufficiently close to a sparse one. 

Our algorithm is a random walk on the angular momentum branching diagram associated with the computation. One distinctive feature of the algorithm is then that it is not limited to the angular momentum and can be extended to other branching diagrams. It will remain efficient as long as the counterparts of the Clebsch-Gordan coefficients remain efficiently computable to high precision and the out-degree of any vertex of the branching diagram will be bounded by a constant (see also the discussion in~\cite{Havlicek18}). One of the interesting cases where our techniques could apply with little adaptation is the case of q-deformations of the $\textsf{SU}(2)$ branching diagrams, applied in the study of topological phases of matter~\cite{Hormozi07, Fern15}. 

Circuits using similar structure but using an $\textsf{SU}(d)$ Schur-Weyl transformation for $d > 2$ were recently applied in study of Boson Sampling with partially distinguishable bosons in the first quantization~\cite{moylett2018quantum}. The possibility of leveraging the simulation techniques proposed here in this context remains open.

\section*{Acknowledgements} V.H. was supported by Keble de Breyne and Clarendon scholarship at the University of Oxford. S.S. was supported the Leverhulme Early Career Fellowship. K.T. acknowledges support from the IBM Research Frontiers Institute. Authors are grateful to Richard Jozsa and Greg Kuperberg for reading the manuscript and acknowledge discussions with Will Kirby, Alex Moylett and Peter Turner that significantly helped to improve this manuscript.

%%% BIBLIOGRAPHY %%%
\bibliographystyle{unsrt}
\bibliography{main}

%%%%%%%%%%%%%%%%%%%%%%%%%%%%%%%%%%%%
%%%%%% APPENDICES START HERE %%%%%%%
%%%%%%%%%%%%%%%%%%%%%%%%%%%%%%%%%%%%
%\pagebreak
\onecolumngrid
\appendix

\section{$A \supseteq B$, then $S_A$ and $S_B$ commute} 
 For sets $A, B$ of qubits, $S_A^2$ and $S_B^2$ commute iff $A$ and $B$ are disjoint or one is subset of the other. Setting $A \supseteq B$, we have:
\begin{align*}
\sum_{k \in A} \vec{S}_k &= \overbrace{\sum_{k\in B} \vec{S}_k 
}^{\vec{\alpha}} + \overbrace{\sum_{k \in A / B } \vec{S}_k }^{\vec{\beta}}.
\end{align*}
Note that: 
\begin{align*}
[\vec{\alpha}, \vec{\beta}]  &= 0, & [\vec{\alpha}^2, \vec{\beta}] &= \vec{0}. \end{align*}
This gives:
\begin{align*}
[ S_A^2, S_B^2 ] &= 2[\vec{\alpha} \cdot \vec{ \beta}, \vec{\alpha}^2] = 2[\vec{\alpha} \, \cdot, \vec{\alpha}^2]   \vec{ \beta} = 0.
\end{align*}
\label{App:com1}

 \section{$Z_A$ and $S_B^2$ commute for $A \supseteq B$}
This can be seen by:
\begin{align*} 8 \, [Z_A, S_B^2 ] = 8 \, [Z_B, S_B^2]  &= \sum_{k,l,m \in B} \left[Z_k, \left( X_l X_m + Y_l Y_m + Z_l Z_m \right)\right] \\
&=   \sum_{k,l,m \in B} [Z_k, X_l] X_m + X_l [Z_k, X_m] + [Z_k, Y_l] Y_m + Y_l[Z_k, Y_m] \\
&= 2i  \sum_{k,l \in B} Y_l X_k + X_l Y_k - Y_l X_k - X_l Y_k = 0. \end{align*}
\label{App:com2}

 \section{Diagrammatic representation of the spin basis states}
\label{Appendix:StateDiagrams}

\begin{tikzpicture}[scale = 1]
    \node[circle, inner sep=2pt] at(-.65, 0.75) {$1$};
    \node[circle, inner sep=2pt] at(.65, 0.75) {$3$};
    \node[circle, inner sep=2pt] at(-0.0, 0.75) {$2$};

    \node at(-.9, -0.17) {$j_{[2]} = 0$};
    \node at(-0, -0.9) {$J = \frac{1}{2}, M = \frac{1}{2}$};
    
    \draw[color=gray,thick] (-0.65, 0.6) -- (-0.3, 0.1); 
    \draw[color=gray,thick] (-0, 0.6) -- (-0.3, 0.1); 
    \draw[color=gray,thick] (0.65, 0.6) -- (0.3, 0.1); 
    \draw[color=gray,thick] (-0.3,0.1)  -- (0, -.3);
    \draw[color=gray,thick] (0.3,0.1) -- (0, -0.3);
    \draw[color=gray,thick] (0, -.3) -- (0, -0.7);
    
    \node at(0,-1.6) {$\Ket{J = \frac{1}{2}, M=\frac{1}{2}, j_{[2]} = 0}$};
    
    \node[circle, fill = black,inner sep=1pt] at (-0.3, 0.1){};
    \node[circle, fill = black,inner sep=1pt] at (0, -0.3){};
\end{tikzpicture}
\hspace{.2cm}
%%% J = 1/2, M = -1/2, j = 0
\begin{tikzpicture}[scale = 1]
    \node[circle, inner sep=2pt] at(-.65, 0.75) {$1$};
    \node[circle, inner sep=2pt] at(.65, 0.75) {$3$};
    \node[circle, inner sep=2pt] at(-0.0, 0.75) {$2$};

    \node at(-.9, -0.17) {$j_{[2]} = 0$};
    \node at(-0, -0.9) {$J = \frac{1}{2}, M = -\frac{1}{2}$};
    
    \draw[color=gray,thick] (-0.65, 0.6) -- (-0.3, 0.1); 
    \draw[color=gray,thick] (-0, 0.6) -- (-0.3, 0.1); 
    \draw[color=gray,thick] (0.65, 0.6) -- (0.3, 0.1); 
    \draw[color=gray,thick] (-0.3,0.1)  -- (0, -.3);
    \draw[color=gray,thick] (0.3,0.1) -- (0, -0.3);
    \draw[color=gray,thick] (0, -.3) -- (0, -0.7);
    
    \node at(0,-1.6) {$\Ket{J = \frac{1}{2}, M=-\frac{1}{2}, j_{[2]} = 0}$};
    
    \node[circle, fill = black,inner sep=1pt] at (-0.3, 0.1){};
    \node[circle, fill = black,inner sep=1pt] at (0, -0.3){};
\end{tikzpicture}
\hspace{.2cm}
%%% J = 1/2, M = 1/2, j = 1
\begin{tikzpicture}[scale = 1]
    \node[circle, inner sep=2pt] at(-.65, 0.75) {$1$};
    \node[circle, inner sep=2pt] at(.65, 0.75) {$3$};
    \node[circle, inner sep=2pt] at(-0.0, 0.75) {$2$};

    \node at(-.9, -0.17) {$j_{[2]} = 1$};
    \node at(-0, -0.9) {$J = \frac{1}{2}, M=\frac{1}{2}$};
    
    \draw[color=gray,thick] (-0.65, 0.6) -- (-0.3, 0.1); 
    \draw[color=gray,thick] (-0, 0.6) -- (-0.3, 0.1); 
    \draw[color=gray,thick] (0.65, 0.6) -- (0.3, 0.1); 
    \draw[color=gray,thick] (-0.3,0.1)  -- (0, -.3);
    \draw[color=gray,thick] (0.3,0.1) -- (0, -0.3);
    \draw[color=gray,thick] (0, -.3) -- (0, -0.7);
    
    \node at(0,-1.6) {$\Ket{J = \frac{1}{2}, M=\frac{1}{2}, j_{[2]} = 1}$};
    \node[circle, fill = black,inner sep=1pt] at (-0.3, 0.1){};
    \node[circle, fill = black,inner sep=1pt] at (0, -0.3){};
\end{tikzpicture}
\hspace{0.2cm}
%%% J = 1/2, M = -1/2, j = 1
\begin{tikzpicture}[scale = 1]
    \node[circle, inner sep=2pt] at(-.65, 0.75) {$1$};
    \node[circle, inner sep=2pt] at(.65, 0.75) {$3$};
    \node[circle, inner sep=2pt] at(-0.0, 0.75) {$2$};
    \node at(-.9, -0.17) {$j_{[2]} = 1$};
    \node at(-0, -0.9) {$J =\frac{1}{2}, M=-\frac{1}{2}$};
    
    \draw[color=gray,thick] (-0.65, 0.6) -- (-0.3, 0.1); 
    \draw[color=gray,thick] (-0, 0.6) -- (-0.3, 0.1); 
    \draw[color=gray,thick] (0.65, 0.6) -- (0.3, 0.1); 
    \draw[color=gray,thick] (-0.3,0.1)  -- (0, -.3);
    \draw[color=gray,thick] (0.3,0.1) -- (0, -0.3);
    \draw[color=gray,thick] (0, -.3) -- (0, -0.7);
    
    \node at(0,-1.6) {$\Ket{J = \frac{1}{2}, M=-\frac{1}{2}, j_{[2]} = 1}$};
    \node[circle, fill = black,inner sep=1pt] at (-0.3, 0.1){};
    \node[circle, fill = black,inner sep=1pt] at (0, -0.3){};
\end{tikzpicture}

\medskip

%%% J = 3/2, M = 3/2, j = 1
\begin{tikzpicture}[scale = 1]
    \node[circle, inner sep=2pt] at(-.65, 0.75) {$1$};
    \node[circle, inner sep=2pt] at(.65, 0.75) {$3$};
    \node[circle, inner sep=2pt] at(-0.0, 0.75) {$2$};
  
    \node at(-.9, -0.17) {$j_{[2]} = 1$};
    \node at(-0, -0.9) {$J = \frac{3}{2}, M = \frac{3}{2}$};
    
    \draw[color=gray,thick] (-0.65, 0.6) -- (-0.3, 0.1); 
    \draw[color=gray,thick] (-0, 0.6) -- (-0.3, 0.1); 
    \draw[color=gray,thick] (0.65, 0.6) -- (0.3, 0.1); 
    \draw[color=gray,thick] (-0.3,0.1)  -- (0, -.3);
    \draw[color=gray,thick] (0.3,0.1) -- (0, -0.3);
    \draw[color=gray,thick] (0, -.3) -- (0, -0.7);
    
    \node at(0,-1.6) {$\Ket{J = \frac{3}{2}, M=\frac{3}{2}, j_{[2]} = 1}$};
    
    \node[circle, fill = black,inner sep=1pt] at (-0.3, 0.1){};
    \node[circle, fill = black,inner sep=1pt] at (0, -0.3){};
\end{tikzpicture}
\hspace{.2cm}
%%% J = 3/2, M = 1/2, j = 1
\begin{tikzpicture}[scale = 1]
    \node[circle, inner sep=2pt] at(-.65, 0.75) {$1$};
    \node[circle, inner sep=2pt] at(.65, 0.75) {$3$};
    \node[circle, inner sep=2pt] at(-0.0, 0.75) {$2$};
 
    \node at(-.9, -0.17) {$j_{[2]} = 1$};
    \node at(-0, -0.9) {$J=\frac{3}{2}, M=\frac{1}{2}$};
    
    \draw[color=gray,thick] (-0.65, 0.6) -- (-0.3, 0.1); 
    \draw[color=gray,thick] (-0, 0.6) -- (-0.3, 0.1); 
    \draw[color=gray,thick] (0.65, 0.6) -- (0.3, 0.1); 
    \draw[color=gray,thick] (-0.3,0.1)  -- (0, -.3);
    \draw[color=gray,thick] (0.3,0.1) -- (0, -0.3);
    \draw[color=gray,thick] (0, -.3) -- (0, -0.7);
    
    \node at(0,-1.6) {$\Ket{J = \frac{3}{2}, M=\frac{1}{2}, j_{[2]} = 1}$};
    
    \node[circle, fill = black,inner sep=1pt] at (-0.3, 0.1){};
    \node[circle, fill = black,inner sep=1pt] at (0, -0.3){};
\end{tikzpicture}
\hspace{.2cm}
%%% J = 3/2, M = -1/2, j = 1
\begin{tikzpicture}[scale = 1]
    \node[circle, inner sep=2pt] at(-.65, 0.75) {$1$};
    \node[circle, inner sep=2pt] at(.65, 0.75) {$3$};
    \node[circle, inner sep=2pt] at(-0.0, 0.75) {$2$};
    \node[circle, fill = black,inner sep=1pt] at (-0.3, 0.1){};
    \node[circle, fill = black,inner sep=1pt] at (0, -0.3){};

    \draw[color=gray,thick] (-0.65, 0.6) -- (-0.3, 0.1); 
    \draw[color=gray,thick] (-0, 0.6) -- (-0.3, 0.1); 
    \draw[color=gray,thick] (0.65, 0.6) -- (0.3, 0.1); 
    \draw[color=gray,thick] (-0.3,0.1)  -- (0, -.3);
    \draw[color=gray,thick] (0.3,0.1) -- (0, -0.3);
    \draw[color=gray,thick] (0, -.3) -- (0, -0.7);
    
    \node at(0,-1.6) {$\Ket{J = \frac{3}{2}, M=-\frac{1}{2}, j_{[2]} = 1}$};
    
    \node at(-.9, -0.17) {$j_{[2]} = 1$};
    \node at(-0, -0.9) {$J=\frac{3}{2}, M=-\frac{1}{2}$};
\end{tikzpicture}
\hspace{0.2cm}
%%% J = 3/2, M = -3/2, j = 1
\begin{tikzpicture}[scale = 1]
    \node[circle, inner sep=2pt] at(-.65, 0.75) {$1$};
    \node[circle, inner sep=2pt] at(.65, 0.75) {$3$};
    \node[circle, inner sep=2pt] at(-0.0, 0.75) {$2$};
    
    \node at(-.9, -0.17) {$j_{[2]} = 1$};
    \node at(-0, -0.9) {$J=\frac{3}{2}, M=-\frac{3}{2}$};
    
    \draw[color=gray,thick] (-0.65, 0.6) -- (-0.3, 0.1); 
    \draw[color=gray,thick] (-0, 0.6) -- (-0.3, 0.1); 
    \draw[color=gray,thick] (0.65, 0.6) -- (0.3, 0.1); 
    \draw[color=gray,thick] (-0.3,0.1)  -- (0, -.3);
    \draw[color=gray,thick] (0.3,0.1) -- (0, -0.3);
    \draw[color=gray,thick] (0, -.3) -- (0, -0.7);
    
    \node[circle, fill = black,inner sep=1pt] at (-0.3, 0.1){};
    \node[circle, fill = black,inner sep=1pt] at (0, -0.3){};
    
    \node at(0,-1.6) {$\Ket{J = \frac{3}{2}, M=-\frac{3}{2}, j_{[2]} = 1}$};
\end{tikzpicture}
\label{Fig:ExampleStates}

\section{Completeness of the sequentially coupled basis}
The argument comes from \cite{Pauncz67}.
Denote the number of paths in $\mathcal{A}_k$  that end at $j_{[k]}$ by $d(j_{[k]})$. It follows from Eq.~\ref{Eq:couplingSeq} that such $j_{[k]}$ can be reached by taking a step in a path $\bm j_{k-1} \in \mathcal{A}_{k-1}$  that ends  either at $j_{[k-1]}+\frac{1}{2}$  or $j_{[k-1]}-\frac{1}{2}$. This gives a recurrence: 
 \[ d(j_{[k]}) = d\left(j_{[k-1]}-\frac{1}{2}\right) + d\left(j_{[k-1]}+\frac{1}{2}\right),\]
which is solved by:
 \begin{align*} d(J) = {n \choose \frac{1}{2}n -J} - {n \choose \frac{1}{2}n -J -1}.\end{align*}
 The $J$ eigenspaces are $2J +1$-degenerate due to possible values of the $M$ number. We then have that:
 \begin{align} \sum_{\bm J \in \mathcal{A}_n} (2J +1) = \sum_{J=0}^{n} (2J +1) d(J) = 2^n, \label{Eq:SumDim} \end{align}
It follows that eigenstates of $\mathcal{S}_n$ span the $n$-qubit Hilbert space. This also implies that there exist exponentially large blocks for fixed $J$ that asymptotically scale as $2^n$, since the summation in Eq.~\ref{Eq:SumDim} runs only over polynomially many $J$.  In particular, this makes the sampling algorithm of Theorem \ref{Theorem:WeakSampling} run in polynomial time.
\label{App:completeness}

\section{Paths to Young Tableaux}
\label{App:PathsToYD}
Here we show that the paths are one to one with the standard Young tableaux on two rows, which we use to give an improved sampling method in Appendix~\ref{App:GNW}.
Let $\bm J \in \mathcal{A}_n$ be a path and let: 
\begin{align*} x = x_1x_2\ldots x_{n-1} \in \lbrace{0,1}\rbrace^{n-1}, \end{align*} 
be its Yamanouchi symbol. The shape of the corresponding standard Young tableau is determined by $J$ and $n$ - it will have $\frac{n}{2} + J$ boxes in the first row and $\frac{n}{2} - J$ boxes in the second row. Write $1$ to the first box in the upper row, then read the Yamanouchi symbol $x$ from left to right. If the $i$-th bit $x_i = 0$, add an element $i+1$ to the leftmost empty box in the \textit{lower} row. Conversely, if $x_i = 1$, add $i+1$ to the leftmost empty box in the \textit{upper} row. The resulting Young tableau is in the standard form (its elements are increasing both along its rows and columns). The elements in each row are increasing by construction. The elements in each column also increase, which can be seen from the property that any prefix of length $m \leq n-1$ of the Yamanouchi bitstring contains at most $\lceil \frac{m}{2}\rceil$ zeroes -- in other words, the upper row will be always filled faster than the lower one. Paths are also \textit{onto} the standard two-row Young tableaux, which can be proved by converting the tableaux to bitstrings by reversing the above algorithm and checking the defining property of the Yamanouchi symbol. 

As an example, take the sequentially coupled basis state on $n=3$ qubits:
\begin{align*}
\Ket{ J = \frac{1}{2}, M = \frac{1}{2},j_{[2]} = 0} = \Ket{ \bm J, M}.
\end{align*}
with $\bm J = \left[ \frac{1}{2} \rightarrow 0 \rightarrow \frac{1}{2} \right]$. The path ends at $J = \frac{1}{2}$, which means that the corresponding Young diagram will have $2$ boxes in the upper and $1$ in the lower rows: 
\begin{align*}
\yng(2,1)
\end{align*}
The path for this state is $\left[\frac{1}{2} \rightarrow 0 \rightarrow \frac{1}{2} \right]$, which gives a Yamanouchi symbol $\searrow\nearrow = 01$. It also gives a prescription to fill the Young diagram by the above algorithm, giving the tableau: 
\begin{align*} 01 \cong \young(13,2), \end{align*} 
so that the quantum state can be equivalently labeled as:
\begin{align*}
\Ket{ M = \frac{1}{2}, \young(13,2)}
\end{align*}
There is a one-to-one correspondence between the semi-standard Young tableaux of the same shape filled with $\uparrow, \downarrow$ and $M$ -- see \cite{Kirby17b} for discussion of this. However, since $M$ and $n$ completely determine the filling in this case, there is no need to use this here. 

\section{Sampling with the Greene-Nijenhuis-Wilf algorithm}
We now describe how to sample the paths with $n$ steps uniformly randomly using the algorithm proposed by Greene, Nijenhuis and Wilf in \cite{Greene79}. First, fix an endpoint of the path by sampling $J$ from the distribution $\Pi(J) = \frac{2J+1}{2^n} d(J)$ where $d(J)$ is the number of paths that end at $J$, as defined in Appendix~\ref{App:completeness}. Take a two-row Young diagram with $\frac{n}{2}+J$ boxes in the upper and $\frac{n}{2}-J$ in the lower row and use the GNW algorithm to uniformly generate a standard Young Tableaux of this shape - every such tableau is sampled with probability $\frac{1}{d(J)}$ and the sampling algorithm runs in $O(n^2)$ time. Convert the Young tableau to the Yamanouchi symbol and the corresponding path $\bm J$ using Appendix~\ref{App:PathsToYD}. Lastly, choose $M \in \lbrace-J, -J +1, \ldots J\rbrace$ uniformly randomly. The probability of choosing a specific $(\bm J, M)$ is then given by:
\begin{align*}
\Pi(J) \frac{1}{(2J+1)d(J)}  &=  \frac{1}{2^n},
\end{align*}
as wanted. The sampling procedure is then the following: s
\begin{theorem}
\label{Theorem:WeakSampling} Assume that $p(\bm J, M)$ is $\epsilon$-approximate $t$-sparse. It can be sampled classically in $\poly(n, \frac{1}{\epsilon}, t)$ time to $6\epsilon$ error in the total variational distance.
\end{theorem}
\begin{proof}
Use the Kushilevitz-Mansour algorithm in Theorem~\ref{Theorem:Kushilevitz-Mansour} with threshold $\theta = \frac{\epsilon}{t}$ to find $L$ and compute $b = \sum_{\bm J \in L; M} \tilde{p}(\bm J, M)$. Flip a coin with a bias $b$.
\begin{itemize}
\item With probability $b$, output a sample drawn from $\tilde{p}(\bm J, M)/b$ for $\bm J \in L$ and corresponding $M$.
\item With probability $1-b$, output $(\bm J, M)$ for $\bm J \not\in L$ uniformly randomly.
\end{itemize}
To sample $(\bm J,M)$ uniformly randomly,  use the above algorithm to uniformly randomly generate a $(\bm J, M)$ and check if $\bm J \not\in L$. If yes, output. If no, sample again.
\end{proof}

\label{App:GNW}

\section{Simplification of the marginal projector}
\label{App:Projector}
The aim of this section is to simplify the marginal projector expression as:
\begin{align*}
\Pi \left({\bm j}\right) &= \sum_{M} \sum_{\bm J \supseteq \bm j} \Ket{\bm J, M} \Bra{\bm J, M} = \sum_{m} \Ket{ \bm j, m} \Bra{\bm j, m},
\end{align*}
for $\bm j \in \mathcal{A}_k$ and $\sum_{\bm J \supseteq \bm j}$ runs over all paths $\bm J \in \mathcal{A}_n$ that contain $\bm j$. The sum $\sum_m$ runs over $m \in \lbrace -j, -j+1,\ldots j \rbrace$.\\

To do so, we repeatedly use the Clebsch-Gordan orthogonality:
\begin{align*} \sum_{JM} C^{JM}_{jm,j'm_2}C^{JM}_{jm',j'm_2'} &= \delta_{m_2,m_2'}\delta_{mm'}. \end{align*} As we study coupling in the sequential basis, we have that:  
\begin{align*} \sum_{JM} C^{JM}_{jm,m_2}C^{JM}_{jm',m_2'} &= \delta_{m_2,m_2'}\delta_{mm'}. \end{align*} 
We have for the projector $\Pi \left({\bm j}\right)$ that:
\begin{align*}
\Pi \left({\bm j}\right) &= \sum_{M} \sum_{\bm J \supseteq \bm j} \Ket{\bm J, M} \Bra{\bm J, M} \\&= \sum_{\bm J_{n-1} \supseteq \bm j} \sum_{J,M}  \sum_{m_n,m_n'}\sum_{ M_{n-1}, M_{n-1}'} C^{J,M}_{J_{n-1} M_{n-1}; m_n} \ket{m_n} \Ket{\bm J_{n-1}, M_{n-1}} \Bra{\bm J_{n-1}, M_{n-1}'} \bra{m_n'}  C^{J,M}_{J_{n-1}, M_{n-1}'; m_n'} ,
\end{align*}
where $\sum_{\bm J_{n-1} \supseteq \bm j}$ runs over all $\bm J_{n-1} \in \mathcal{A}_{n-1}$ that contain $\bm j$. The $\sum_J$ runs over all allowed $J$. The Clebsch-Gordan coefficients are only non-zero for $J = \left| J_{n-1} \pm \frac{1}{2} \right|$. Using the CG orthogonality, this evaluates to:
\begin{align*}
\Pi \left({\bm j}\right) &= \sum_{\bm J_{n-1} \supseteq \bm j} \sum_{m_n,m_n'} \sum_{ M_{n-1}, M_{n-1}'} \delta_{m_n, m_n'} \delta_{M_{n-1}, M_{n-1}'} \ket{m_n} \Ket{\bm J_{n-1}, M_{n-1}} \Bra{\bm J_{n-1}, M_{n-1}'} \bra{m_n'}  \\
&= \sum_{\bm J_{n-1} \supseteq \bm j} \sum_{M_{n-1}} \underbrace{\sum_{m_n} \ket{m_n}\bra{m_n}}_{\mathbb{I}_{2\times 2}} \otimes \Ket{\bm J_{n-1}, M_{n-1}} \Bra{\bm J_{n-1}, M_{n-1}} = \sum_{M_{n-1}} \sum_{\bm J_{n-1} \supseteq \bm j}  \Ket{\bm J_{n-1}, M_{n-1}} \Bra{\bm J_{n-1}, M_{n-1}}.
\end{align*}
This has the same form as the initial expression, but the projector is now defined by summing over paths with $n-1$ steps. It is possible to continue recursively and write:
\begin{align*}
\Pi \left({\bm j}\right) &= \sum_{M_{n-i}}\sum_{\bm J_{n-i} \supseteq \bm j} \Ket{\bm J_{n-i}, M_{n-i}} \Bra{\bm J_{n-i}, M_{n-i}},
\end{align*}
for any integer $0 \leq i \leq n-k$. For $i = n-k$, one obtains that: 
\begin{align*}
\Pi \left({\bm j}\right) &= \sum_{M_{k}}\sum_{\bm J_{k} \supseteq \bm j} \Ket{\bm J_{k}, M_{k}} \Bra{\bm J_{k}, M_{k}}.
\end{align*}
Since $\bm j \in \mathcal{A}_k$, there is only one path contributing to $\sum_{\bm J_{k} \supseteq \bm j}$, the path $\bm j$ itself. It follows that: 
\begin{align*}
\Pi \left({\bm j}\right) &= \sum_{M_{k}} \Ket{\bm j, M_{k}} \Bra{\bm j, M_{k}}.
\end{align*}
We can write: 
\begin{align*}
\Pi \left({\bm j}\right) &= \sum_{m} \Ket{\bm j, m} \Bra{\bm j, m}.
\end{align*}
where the final summation runs over $m \in \lbrace -j, -j+1,\ldots j \rbrace$.

\end{document}